\documentclass[oneside,article]{memoir}
\usepackage{amsmath}         
\usepackage{amsthm,thmtools,thm-restate} 
\usepackage{enumitem}        

\usepackage{rotating}

\usepackage[final]{listings}
\usepackage{bbm}
\usepackage{xcolor}

\definecolor{codegreen}{rgb}{0,0.6,0}
\definecolor{codegray}{rgb}{0.5,0.5,0.5}
\definecolor{codepurple}{rgb}{0.58,0,0.82}
\definecolor{backcolour}{rgb}{0.95,0.95,0.92}

\lstdefinestyle{mystyle}{
    backgroundcolor=\color{backcolour},   
    commentstyle=\color{codegreen},
    keywordstyle=\color{magenta},
    numberstyle=\tiny\color{codegray},
    stringstyle=\color{codepurple},
    basicstyle=\ttfamily\footnotesize,
    breakatwhitespace=false,         
    breaklines=true,                 
    captionpos=b,                    
    keepspaces=true,                 
    numbers=left,                    
    numbersep=5pt,                  
    showspaces=false,                
    showstringspaces=false,
    showtabs=false,                  
    tabsize=2
}

\usepackage{algorithm}
\usepackage{algpseudocode}
\usepackage{algorithmicx}
\usepackage{threeparttable}
\usepackage{booktabs}
\usepackage{float}

\usepackage{pgfplots}
\pgfplotsset{compat=1.17}
\usepgfplotslibrary{groupplots}
\definecolor{nvcol}{RGB}{170,60,50}
\definecolor{ngcol}{RGB}{215,140,40}
\definecolor{nqcol}{RGB}{45,110,165}
\definecolor{npcol}{RGB}{40,130,90}

\usepackage{tikz}
  \usetikzlibrary{matrix,arrows,positioning}
\usepackage{amsfonts}
\usepackage{parskip}
\usepackage{mathtools}

\usepackage{tikz-cd}
\usepackage{float} 
\usepackage{caption, subcaption} 

\usepackage[nobysame,alphabetic,initials]{amsrefs} 
\usepackage[hidelinks]{hyperref}                   
\usepackage[capitalize,nosort]{cleveref}           

\usepackage{amssymb}
\usepackage[mathscr]{euscript}
\usepackage{relsize}
\usepackage{stmaryrd}
\usepackage{stackengine}

\usepackage{indentfirst} 
\usepackage{multicol}

\usepackage[draft]{fixme} 
\fxsetup{
  status=draft,
  layout=inline
}
\usepackage{graphicx}
\usepackage{subcaption}

\usepackage{float} 

\isopage[12]

\setlrmargins{*}{*}{1}
\checkandfixthelayout

\counterwithout{section}{chapter}
\usepackage{appendix}

  \newcommand{\id}[1][]{\operatorname{id}_{#1}}
  
  \declaretheorem[style=definition,within=section]{definition}
  \declaretheorem[style=definition,numberlike=definition]{example}

  \declaretheorem[style=plain,numberlike=definition]{proposition}
  \declaretheorem[style=plain,numberlike=definition]{theorem}
  \declaretheorem[style=plain,numberlike=definition]{conjecture}

  \declaretheorem[style=plain,numbered=no,name=Theorem]{theorem*}

  \Crefname{corollary}{Corollary}{Corollaries}
  \Crefname{definition}{Definition}{Definitions}
  \Crefname{lemma}{Lemma}{Lemmas}
  \Crefname{proposition}{Proposition}{Propositions}
  \Crefname{remark}{Remark}{Remarks}
  \Crefname{theorem}{Theorem}{Theorems}
  \Crefname{notation}{Notation}{Notations}
  \Crefname{conjecture}{Conjecture}{Conjectures}

  \newlist{axioms}{enumerate}{1}
  \Crefname{axiomsi}{}{}

  \newenvironment{tikzeq*}
  {
    \begingroup
    \begin{equation*}
    \begin{tikzpicture}[baseline=(current bounding box.center)]
  }
  {
    \end{tikzpicture}
    \end{equation*}
    \endgroup
    \ignorespacesafterend
  }

  \tikzset
  {
    diagram/.style=
    {
      matrix of math nodes,
      column sep={4.3em,between origins},
      row sep={4em,between origins},
      text height=1.5ex,
      text depth=.25ex
    },
    over/.style={preaction={draw=white,-,line width=6pt}},
    every to/.style={font=\footnotesize},
    inj/.style={right hook->},
    surj/.style={-{Latex[open]}},
    cof/.style={>->},
    fib/.style={->>},
  }

  \DeclareFontFamily{U}{mathx}{\hyphenchar\font45}

  \DeclareFontShape{U}{mathx}{m}{n}{
    <5> <6> <7> <8> <9> <10>
    <10.95> <12> <14.4> <17.28> <20.74> <24.88>
    mathx10}{}

  \DeclareSymbolFont{mathx}{U}{mathx}{m}{n}

  \DeclareFontFamily{U}{mathb}{\hyphenchar\font45}

  \DeclareFontShape{U}{mathb}{m}{n}{
    <5> <6> <7> <8> <9> <10>
    <10.95> <12> <14.4> <17.28> <20.74> <24.88>
    mathb10}{}

  \DeclareSymbolFont{mathb}{U}{mathb}{m}{n}

  \DeclareMathAccent{\widebar}{0}{mathx}{"73}

  \DeclareMathSymbol{\Rsh}{\mathrel}{mathb}{"E9}

  \DeclareFontFamily{U}{MnSymbolA}{}

  \DeclareFontShape{U}{MnSymbolA}{m}{n}{
    <-6> MnSymbolA5
    <6-7> MnSymbolA6
    <7-8> MnSymbolA7
    <8-9> MnSymbolA8
    <9-10> MnSymbolA9
    <10-12> MnSymbolA10
    <12-> MnSymbolA12}{}

  \DeclareSymbolFont{MnSyA}{U}{MnSymbolA}{m}{n}

  \DeclareMathSymbol{\twoheaddownarrow}{\mathrel}{MnSyA}{27}

  \newcommand{\MSC}[1]{%
    \let\thempfn\relax
    \footnotetext[0]{2020 Mathematics Subject Classification: #1.}
  }

  \newcommand{\Ker}{\textsf{Ker}}

\tikzstyle{vertex}=[circle, draw, minimum size=7pt, inner sep=0pt]

\newcommand{\im}{\mathsf{Im}} 

\DeclareFontFamily{U}{dmjhira}{}
\DeclareFontShape{U}{dmjhira}{m}{n}{ <-> dmjhira }{}

\author{Krzysztof Kapulkin \and Nathan Kershaw} 
\title{Faster computations of discrete homology}
\date{\today}

\begin{document}

  \maketitle

\begin{abstract}
  Machine computation of the discrete homology of graphs has stopped at degree two.
  We present an algorithm that reaches degree four.
  It generates the singular cubes inductively, pairing cubes one degree down instead of filtering all set maps; quotients the chain modules by the hyperoctahedral group action, over a field of sufficiently large characteristic; and shrinks the graph beforehand using homotopy invariance.
  The fourth homology group of the five-cycle, previously beyond the reach of computation, is computed in under two days.
\end{abstract}



\section*{Introduction}

\emph{Discrete homology} is an invariant of graphs, introduced by Barcelo, Capraro,
and White \cite{barcelo-capraro-white} in the setting of finite metric spaces and
studied in the context of graphs by Barcelo, Greene, Jarrah, and Welker
\cite{barcelo-greene-jarrah-welker:comparison,barcelo-greene-jarrah-welker:connections,barcelo-greene-jarrah-welker:vanishing}.
It is the homology of a chain complex whose degree-$n$ generators are the
non-degenerate graph maps out of the hypercube graph $Q^n$, and it belongs to a
wider theory of combinatorial invariants known as discrete homotopy theory
\cite{babson-barcelo-longueville-laubenbacher,barcelo-laubenbacher,carranza-kapulkin:cubical-graphs}.
Where the homology of the clique complex fills in the triangles of a graph,
discrete homology fills in both the triangles and the squares, and this small
difference is what gives it its distinctive character: it is sensitive to
combinatorial rather than topological holes, and it sees a cycle as a hole only
once that cycle has length at least five.

That sensitivity is what makes the invariant useful, and it has been put to work in
several directions.
In matroid theory, Maurer \cite{maurer:basis-graphs} showed that the basis graph of
a matroid has trivial discrete fundamental group; since the first discrete homology
group is the abelianization of that group
\cite[Thm.~4.1]{barcelo-capraro-white}, a non-vanishing $\mathcal{H}_1$ certifies
that a graph is \emph{not} the basis graph of any matroid, and this is a condition
one can actually check by computation.
In the theory of subspace arrangements, a result of Babson and Bj\"{o}rner
identifies the fundamental group of the complement of the $k$-equal arrangement with
a discrete fundamental group of the order complex of a Boolean lattice
\cite{barcelo-severs-white}, so discrete invariants of an explicitly presented graph
compute a topological invariant of an arrangement.
Further connections, to $k$-connectivity of graphs and to $q$-analysis of simplicial
complexes, are surveyed by Barcelo and Laubenbacher \cite{barcelo-laubenbacher}.

More recently the invariant has found a use in data analysis.
Vietoris--Rips persistence is reliable on metric data but degrades badly when the
triangle inequality fails, as it does for similarity measures such as correlation
between time series; the four-cycles that noise creates most easily are precisely
the ones the clique complex leaves unfilled, and they generate spurious classes that
survive across wide ranges of the filtration parameter.
Discrete homology fills those cycles by construction, so the spurious generators
never appear.
In \cite{kapulkin-kershaw:tda} we gave evidence for this.
On a noisy sample of a circle, the barcode computed with discrete homology carried
four short bars to the clique complex's seventeen; and across a thousand randomized
perturbations of a circle, the persistence diagram obtained from discrete homology
was the closer of the two to that of the unperturbed circle, in bottleneck distance,
in $943$ trials.

So there are good reasons to want these groups.
The difficulty is that discrete homology is a \emph{singular} theory, and singular
theories are expensive.
The degree-$n$ chain module is free on the non-degenerate maps $Q^n \to G$, and the
obvious way to find these examines all $|G_V|^{2^{n}}$ set maps, a count that is
doubly exponential in the degree.
Even the modules themselves are enormous.
The Greene sphere of \cref{ex:graphs} has ten vertices, and admits $442$ singular
$2$-cubes, $22\,762$ singular $3$-cubes, over $2 \times 10^{7}$ singular $4$-cubes,
and more than $6.4 \times 10^{12}$ singular $5$-cubes --- the last of which are
exactly what a computation of $\mathcal{H}_4$ requires.
This is why machine computation has stopped at degree $2$, and why the fourth
homology group of the $5$-cycle --- a graph on five vertices, whose homology is
known in all degrees by a theoretical argument --- was nonetheless regarded as beyond
the reach of computation \cite{barcelo-greene-jarrah-welker:comparison}.

This scale also explains why the standard machinery of computational topology does
not transfer.
That machinery is highly developed.
The persistence algorithm of \cite{zomorodian-carlsson:computing-persistent-homology}
reduces the boundary matrix by column operations, and the clearing and twist
optimizations \cite{chen-kerber:twist,bauer-kerber-reininghaus:chunks}, together with
the use of cohomology, make it fast enough in practice that implementations such as
\texttt{Ripser} \cite{bauer:ripser} handle Vietoris--Rips filtrations on data sets far
larger than anything considered here.
Running alongside it is a body of work on shrinking the complex before reducing it:
discrete Morse theory adapted to filtrations \cite{mischaikow-nanda:morse},
coreduction \cite{mrozek-batko:coreduction}, and edge collapse for flag complexes
\cite{boissonnat-siddharth-edge-collapse}.

What all of this shares is the assumption that the complex can be built and then
simplified.
It operates \emph{on} the complex, and it is enormously effective there.
Here the complex is the problem.
A method that begins by enumerating the cubes has already lost, and accordingly every
improvement in this paper is a way of not enumerating them --- of generating the cubes
that matter, in a form that carries the information needed later, and of never
materializing the ones that do not.
Algebraic Morse theory \cite{kozlov-algebraic-morse} is the one item on the list that
is not obviously ruled out, since it applies to any based chain complex rather than to
a cell structure; but a Morse matching is still a structure on the complex, and
whether one can be produced alongside a generation scheme, rather than after it, we
leave open.

The obstruction is not peculiar to graphs.
Essentially the same theory is studied in digital topology, where the objects are
digital images rather than abstract graphs \cite{jamil-ali:digital}, and there too the
cost of the singular construction is the recognized barrier to using it.

Closer to hand are the neighbouring homology theories of graphs, for which the
computational problem has been attacked directly.
Chowdhury and M\'{e}moli \cite{chowdhury-memoli:path} give an algorithm for path
homology and its persistent version, and Dey, Li, and Wang
\cite{dey-li-wang:path} improve it substantially in degree $1$ by identifying
structures specific to that degree.
The contrast is instructive.
Their gains come from understanding what a path-homology class in a fixed low degree
looks like, which is what makes a specialized algorithm possible; ours come from
attacking the enumeration itself, which is what the singular construction forces and
what allows us to move in degree rather than within it.

Our contribution is an algorithm that computes degree-$3$ groups for every graph in
our test set and the degree-$4$ group of the $5$-cycle, the latter in roughly $39$
hours.
Measured against the naive algorithm, the speedups run from $10^3$ to $10^{12}$ in
degree $2$ and from $10^8$ to $10^{25}$ in degree $3$.
Three ideas account for this, and they are not independent: the first makes the
other two possible.

The first is a way of generating cubes inductively.
Two maps $f, g \colon Q^{n-1} \to G$ pair to a map $Q^n \to G$ exactly when
$f(v) \sim g(v)$ for every vertex $v$, so the $n$-cubes can be assembled from the
$(n-1)$-cubes without ever examining a set map that is not a graph map
(\cref{sec:generation}).
Because a cube now arrives together with its two top faces, degeneracy data can be
carried along as it is built rather than recovered afterwards, and the remaining
faces can be read off by index arithmetic instead of precomposition
(\cref{sec:degeneracies}).

The second is a quotient.
The hypercube graph carries an action of the hyperoctahedral group
$R_n = S_n \ltimes (\mathbb{Z}_2)^n$, and by a recent result of Greene, Welker, and
Wille \cite{greene-welker-wille} one may quotient the chain complex by the induced
action before taking homology, provided the characteristic does not divide $(n+1)!$.
This removes a factor of up to $n! \cdot 2^n$ from the degree-$n$ module, and the
saving grows quickly with the degree.
It is the inductive generation that makes it usable: the group action can be tracked
through the construction, so equivalence classes are paired directly and the
individual cubes in them are never stored (\cref{sec:quotient}).

The third is to shrink the graph before computing anything.
Discrete homology is invariant under a combinatorial notion of homotopy, and we use
this to delete vertices (\cref{sec:preprocessing}).
The vertices we remove are the \emph{dominated} vertices of
\cite{barmak-et-al-dominated}, and the reduction is the discrete-homotopy analogue
of the edge collapse used for flag complexes
\cite{boissonnat-siddharth-edge-collapse}; the combinatorial condition is identical,
though the invariance theorem behind it is not.
Detecting all such vertices costs one pass over the vertex set, and since the cost of
everything downstream grows so steeply in $|G_V|$, the reduction is worth attempting
unconditionally.

We should be explicit about what these improvements do not do.
They attack the growth in the degree $n$, and on that they are effective.
They do little about the growth in the size of the graph, which remains the
principal obstacle: the complete graph $K_{10}$ is the hardest of our examples in
every degree, for the simple reason that on a complete graph every set map is a
graph map and there is nothing for the generation step to prune.
One further constraint is worth naming, since it rules out an otherwise attractive
option.
Reduction of the boundary matrix over $\mathbb{F}_2$ is the fastest route to ranks in
the persistent setting \cite{chen-kerber:twist,bauer:ripser}, but the quotient above
requires characteristic different from $2$, so the two cannot be combined; we have
chosen the quotient, which is worth considerably more.

Our code is available at
\url{https://github.com/nkershaw01/Discrete_Cubical_Homology};
the implementation and experimental setup are described in \cref{sec:experiments}.
A persistent version of these computations, for filtered graphs, is the subject of
forthcoming work.

The paper is organized as follows.
\cref{sec:homology} reviews discrete homology, and \cref{sec:naive} describes the
naive algorithm that serves as our baseline.
\cref{sec:generation} presents the inductive generation of cubes and
\cref{sec:degeneracies} the improvements it makes possible, after which
\cref{sec:quotient} develops the quotient by the hyperoctahedral group and
\cref{sec:preprocessing} treats preprocessing.
\cref{sec:experiments} reports the experiments, including the effect of
parallelization, and \cref{sec:summary} summarizes and poses open problems.

\textbf{Acknowledgements.}
First and foremost, we thank Daniel Carranza for helpful conversations regarding
this work, in particular those leading to the insights of \cref{sec:generation}.
We are grateful to Curtis Greene, Volkmar Welker, and Georg Wille for sharing early
drafts of their work on homology of cubical sets with symmetries and reversals
\cite{greene-welker-wille}.
We thank Jacob Ender for help with parallelizing the algorithm.
Lastly, we acknowledge the opportunities extended to us by the Digital Research
Alliance of Canada for running large computations.

\section{Discrete homology} \label{sec:homology}

We begin by reviewing the definition of discrete homology.
The notion was first introduced in \cite{barcelo-capraro-white} in the context of finite metric spaces, and our presentation follows that of \cite{barcelo-greene-jarrah-welker:comparison}, where it was first recorded in the context of graphs.

\begin{definition}  \label{def:graphs} \leavevmode
    \begin{enumerate}
        \item A \textit{graph} $G$ is a set equipped with a symmetric and reflexive relation.
        We write $G_V$ for the underlying set, called the \textit{vertex set} and $G_E \subseteq G_V \times G_V$ to be the relation, which we call the \textit{edge set}.
        We write $v \sim w$ to mean $(v,w) \in G_E$.
        \item A \textit{graph map} $f \colon G \to H$ from a graph $G$ to $H$ is a set map $f \colon G_V \to H_V$ that preserves the relation.
    \end{enumerate}
\end{definition}

Our graphs are thus simple and undirected, and reflexivity gives every vertex a unique loop, which allows a graph map to contract an edge to a vertex.
\begin{example} \label{ex:graphs} ${}$

\begin{enumerate}
    \item The graph $I_n$ has $n+1$ vertices, labelled $0, \dots , n$ with an edge between vertices $i$, $i+1$ for $0 \leq i \leq n-1$.
    
    \item The \textit{cycle graph on $n$ vertices}, $C_n$, is the graph $I_n$ quotient by the relation setting $0=n$.

    \item The \textit{complete graph on $n$ vertices} is a graph with $n \in \mathbb{N}$ vertices, and an edge between every pair of vertices.

    \item The \textit{star-shaped $C_5$} graph, $C_5^{star}$ is the graph $C_5$, with a triangle ``glued in'' at every edge, as depicted below.
    \begin{figure}[htbp]
    \centering
    \begin{subfigure}[b]{0.32\textwidth}
        \centering
        \begin{tikzpicture}
            \node[circle, draw, inner sep=2pt, label=above:0] (0) at (0,0) {};
            \node[circle, draw, inner sep=2pt, label=above:1] (1) at (1,0) {};
            \node[circle, draw, inner sep=2pt, label=above:2] (2) at (2,0) {};
            \node[circle, draw, inner sep=2pt, label=above:3] (3) at (3,0) {};
            
            \draw (0) -- (1);
            \draw (1) -- (2);
            \draw (2) -- (3);
        \end{tikzpicture}
        \caption{The graph $I_3$}
    \end{subfigure}
    \hfill
    \begin{subfigure}[b]{0.32\textwidth}
        \centering
        \begin{tikzpicture}
            \node[circle, draw, inner sep=2pt, label=above:0] (0) at (0,0.78) {};
            \node[circle, draw, inner sep=2pt, label=right:1] (1) at (0.95,1.31) {};
            \node[circle, draw, inner sep=2pt, label=above:2] (2) at (0.59,2.19) {};
            \node[circle, draw, inner sep=2pt, label=above:3] (3) at (-0.59,2.19) {};
            \node[circle, draw, inner sep=2pt, label=left:4] (4) at (-0.95,1.31) {};
            
            \draw (0) -- (1);
            \draw (1) -- (2);
            \draw (2) -- (3);
            \draw (3) -- (4);
            \draw (4) -- (0);
        \end{tikzpicture}
        \caption{The graph $C_5$}
    \end{subfigure}
    \hfill
    \begin{subfigure}[b]{0.32\textwidth}
        \centering 
        \begin{tikzpicture}
            \node[circle, draw, inner sep=2pt] (0) at (0,0.78) {};
            \node[circle, draw, inner sep=2pt] (1) at (0.95,1.31) {};
            \node[circle, draw, inner sep=2pt] (2) at (0.59,2.19) {};
            \node[circle, draw, inner sep=2pt] (3) at (-0.59,2.19) {};
            \node[circle, draw, inner sep=2pt] (4) at (-0.95,1.31) {};
            
            \node[circle, draw, inner sep=2pt] (5) at (-0.76, 0.6) {} ;
            \node[circle, draw, inner sep=2pt] (6) at (0.76, 0.6) {} ;
            \node[circle, draw, inner sep=2pt] (7) at (-1.32, 2) {} ;
            \node[circle, draw, inner sep=2pt] (8) at (1.32, 2) {} ;
            \node[circle, draw, inner sep=2pt] (9) at (0, 2.8) {} ;
            
            \draw (0) -- (1);
            \draw (1) -- (2);
            \draw (2) -- (3);
            \draw (3) -- (4);
            \draw (4) -- (0);

            \draw (5) -- (0);
            \draw (5) -- (4);

            \draw (6) -- (0);
            \draw (6) -- (1);

            \draw (7) -- (4);
            \draw (7) -- (3);

            \draw (8) -- (1);
            \draw (8) -- (2);

            \draw (9) -- (3);
            \draw (9) -- (2);
        \end{tikzpicture}
        \caption{The graph $C_5^{star}$}
    \end{subfigure}
    \caption{The graphs $I_3$, $C_5$, and $C_5^{star}$}
    \label{fig:combined_graphs}
\end{figure}
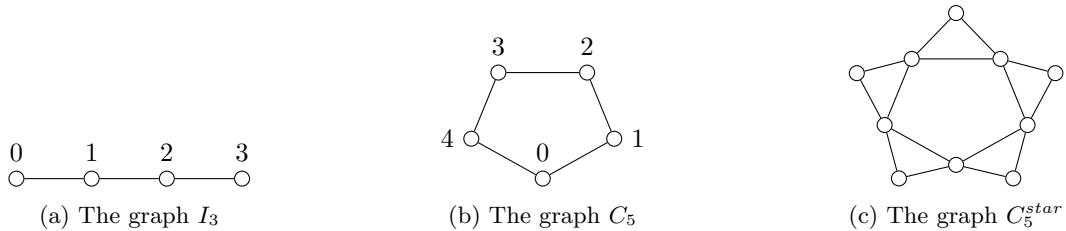

\item The Greene sphere, $G^{sph}$, depicted below.

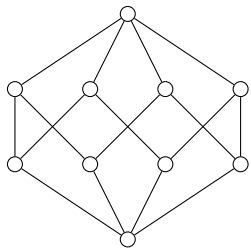
\begin{figure}[H]
    \centering
    \begin{tikzpicture}[scale=1]
        \node[circle, draw, inner sep=2pt] (A) at (0,4) {};
        
        \node[circle, draw, inner sep=2pt] (B1) at (-1.5,3) {};
        \node[circle, draw, inner sep=2pt] (B2) at (-0.5,3) {};
        \node[circle, draw, inner sep=2pt] (B3) at (0.5,3) {};
        \node[circle, draw, inner sep=2pt] (B4) at (1.5,3) {};
        
        \node[circle, draw, inner sep=2pt] (C1) at (-1.5,2) {};
        \node[circle, draw, inner sep=2pt] (C2) at (-0.5,2) {};
        \node[circle, draw, inner sep=2pt] (C3) at (0.5,2) {};
        \node[circle, draw, inner sep=2pt] (C4) at (1.5,2) {};
        
        \node[circle, draw, inner sep=2pt] (D) at (0,1) {};
        
        \draw (A) -- (B1);
        \draw (A) -- (B2);
        \draw (A) -- (B3);
        \draw (A) -- (B4);
        
        \draw (B1) -- (C1);
        \draw (B1) -- (C2);
        \draw (B2) -- (C1);
        \draw (B2) -- (C3);
        \draw (B3) -- (C2);
        \draw (B3) -- (C4);
        \draw (B4) -- (C3);
        \draw (B4) -- (C4);
        
        \draw (C1) -- (D);
        \draw (C2) -- (D);
        \draw (C3) -- (D);
        \draw (C4) -- (D);
    \end{tikzpicture}
    \caption{The Greene sphere}
\end{figure}
\end{enumerate}
\end{example}

Among the many notions of graph product \cite{imrich-klavzar}, the one we need is the box product, also called the cartesian product.

\begin{definition} \label{def:box_prod}
    Given two graphs $G$ and $H$, their \textit{box product}, denoted $G \square H$, is the graph with vertex set $G_V \times H_V$, and such that there is an edge between vertices $(v,w)$ and $(v',w')$ if either $v=v'$ and $w \sim w'$, or $v \sim v'$ and $w = w'$. 
\end{definition}

This gives our last benchmark example.

\begin{example}
The \emph{discrete $3$-torus} $T^3$ is the graph obtained by taking the quotient of $I_5^{\square 3}$ relating vertices $(0,i,j) = (5,i,j)$, $(i,0,j) = (i,5,j)$, and $(i,j,0) = (i,j,5)$.
\end{example}
With $125$ vertices, $T^3$ is the largest of our test graphs.

\begin{figure}[htbp]
    \centering
    \begin{subfigure}[b]{0.32\textwidth}
        \centering
        \begin{tikzpicture}
            \node[circle, draw, inner sep=2pt] (0) at (0,0) {};
            \node[circle, draw, inner sep=2pt] (1) at (1,0) {};

            \draw (0) -- (1);
            
        \end{tikzpicture}
        \caption{$I_1$}
    \end{subfigure}
    \hfill
    \begin{subfigure}[b]{0.32\textwidth}
        \centering
        \begin{tikzpicture}
            \node[circle, draw, inner sep=2pt] (0) at (0,0) {};
            \node[circle, draw, inner sep=2pt] (1) at (0,1) {};
            \node[circle, draw, inner sep=2pt] (2) at (1,0) {};
            \node[circle, draw, inner sep=2pt] (3) at (1,1) {};

            \draw (0) -- (1);
            \draw (1) -- (3);
            \draw (2) -- (3);
            \draw (2) -- (0);

        \end{tikzpicture}
        \caption{$I_1^{\square 2}$}
    \end{subfigure}
    \hfill
    \begin{subfigure}[b]{0.32\textwidth}
        \centering 
        \begin{tikzpicture}
            \node[circle, draw, inner sep=2pt] (0) at (0,0) {};
            \node[circle, draw, inner sep=2pt] (1) at (0,1) {};
            \node[circle, draw, inner sep=2pt] (2) at (1,0) {};
            \node[circle, draw, inner sep=2pt] (3) at (1,1) {};

            \node[circle, draw, inner sep=2pt] (4) at (0.5,0.5) {};
            \node[circle, draw, inner sep=2pt] (5) at (0.5,1.5) {};
            \node[circle, draw, inner sep=2pt] (6) at (1.5,0.5) {};
            \node[circle, draw, inner sep=2pt] (7) at (1.5,1.5) {};

            \draw (0) -- (1);
            \draw (1) -- (3);
            \draw (2) -- (3);
            \draw (2) -- (0);

            \draw (4) -- (5);
            \draw (5) -- (7);
            \draw (6) -- (7);
            \draw (6) -- (4);

            \draw (0) -- (4);
            \draw (1) -- (5);
            \draw (2) -- (6);
            \draw (3) -- (7);
            
        \end{tikzpicture}
        \caption{$I_1^{\square 3}$}
    \end{subfigure}
    \caption{The graphs $I_1$, $I_1^{\square 2}$, and $I_1^{\square 3}$}
    \label{fig:box-powers}
\end{figure}
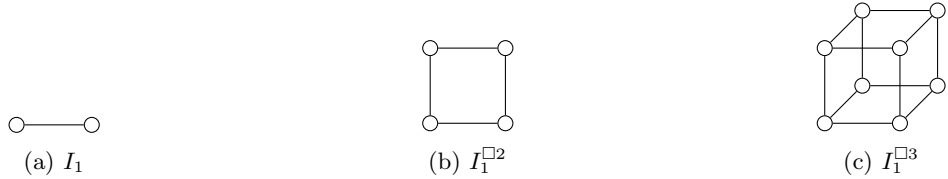

We can now define discrete homology. 

\begin{definition} \label{def:n-cube}
    The \textit{discrete $n$-cube} (or the \emph{hypercube graph}), denoted $I_1^{\square n}$ is the iterated box product of $I_1$ with itself $n$ times.
    We label the vertices $(x_1,\dots,x_n)$ for $x_i \in \{0,1\}$. 
\end{definition}

The name \emph{hypercube graph} and the notation $Q^n$, used in the Introduction, might be more familiar to those with combinatorics background.
We would typically prefer the name \emph{discrete $n$-cube} because of our intended application.

\begin{definition} \label{def:sing_n-cube}
    Let $G$ be a graph. Then a \textit{singular $n$-cube} in $G$ is a map $A \colon I_1^{\square n} \to G$. 
\end{definition}

We often call these simply $n$-cubes in $G$.

\begin{definition} \label{def:face_maps}
    Let $G$ be a graph and $A$ a singular $n$-cube for $n \geq 1$. Then for $1 \leq i \leq n$, we define the \textit{positive and negative face maps} $\delta_i^- A$ and $\delta_i^+ A$ to be the singular $(n-1)$-cubes given by:
    \[ \delta_i^+ A (x_1, \dots, x_{n-1}) := A(x_1, \dots, x_{i-1}, 1, x_{i}, \dots, x_{n-1}), \]
    \[ \delta_i^- A (x_1, \dots, x_{n-1}) := A(x_1, \dots, x_{i-1}, 0, x_{i}, \dots, x_{n-1}). \]

    If $\delta_i^- A = \delta_i^+ A$ for some $i$, we say that $A$ is \textit{degenerate}. Otherwise, we say that $A$ is \emph{non-degenerate}.
    In particular, all 0-cubes are non-degenerate. 
\end{definition}

For the remainder of the section, fix a commutative ring $R$ with unity and a graph $G$. 

\begin{definition}
    For each $n \geq 0$, define $\mathcal{L}_n(G)$ to be the free $R$ module generated by all singular $n$-cubes in $G$. Let $\mathcal{D}_n(G)$ be the submodule generated by the degenerate $n$-cubes in $G$, and set $\mathcal{C}_n(G) = \mathcal{L}_n(G)/\mathcal{D}_n(G)$. 
\end{definition}
The quotient $\mathcal{C}_n(G)$ is free on the set of non-degenerate $n$-cubes in $G$, and we identify it with that free module throughout. 

\begin{definition}
    Given a singular $n$-cube $A$ for $n \geq 1$, define the \textit{boundary} of $A$ to be the formal sum

    $$\partial_n(A) = \sum_{i=1}^n(-1)^i(\delta_i^-A -\delta_i^+ A)$$
\end{definition}

Extending linearly gives us a map $\partial_n \colon \mathcal{L}_n(G) \to \mathcal{L}_{n-1}(G)$. One can show that for $n \geq 0$, we have that $\partial_n[\mathcal{D}_n(G)] \subseteq \mathcal{D}_{n-1}(G)$ and that $\partial_n \circ \partial_{n+1} = 0$ \cite{barcelo-capraro-white}. 
Setting $\mathcal{C}_{-1}(G) = 0$ and $\partial_0$ to be the trivial map, we have a chain complex $\mathcal{C}_\bullet(G) = (\mathcal{C}_\bullet, \partial_\bullet)$ of free $R$-modules. 

\begin{definition}
    For $n \geq 0$, the \textit{$n$-th discrete homology group} of $G$ is $\mathcal{H}_n(G) = \Ker \partial_n / \im \partial_{n+1}$.
\end{definition}

We record the discrete homology of some of the graphs of \cref{ex:graphs}; these are the examples used for benchmarking in \cref{sec:experiments}.

\begin{example} \label{ex:homology_groups}
    The discrete homology groups of $C_5$, $G^{sph}$, and $T^3$ are as follows.

    \begin{itemize}
        \item $\mathcal{H}_n(C_5) = 
        \begin{cases}
            R & n=0,1 \\
            0 & n \geq 2
        \end{cases}$

        \item $\mathcal{H}_n(G^{sph}) = 
        \begin{cases}
            R & n=0,2 \\
            0 & n=1, 3 \\
            \text{unknown} & n\geq 4
        \end{cases}$ 

        \item $\mathcal{H}_n(T^3) = 
        \begin{cases}
            R & n=0,3 \\
            R^3 & n=1,2 \\
            \text{unknown} & n\geq 4
        \end{cases}$
    \end{itemize}
\end{example}
\section{The naive algorithm} \label{sec:naive}

We now describe the ``naive algorithm'' for computing discrete homology,
in which every step follows directly from the definitions.
It is slow, and serves as the baseline against which the rest of the paper is
measured.
From this point on we compute over a field $F$.

There are four steps: generating the graph maps, removing the degeneracies,
building the boundary matrices, and computing their ranks.
Rank computation over a field is well studied
\cite{wiedemann:matrix,coppersmith:matrix,turner:matrix,dumas-giorgi-pernet:matrix,dumas-pernet:matrix,jeannerod-pernet-storjohann:matrix},
and we do not treat it here.
As it turns out, it is not the bottleneck at any stage (\cref{sec:experiments}).

\paragraph{Generating the maps.}
Fix a graph $G$ and $n \geq 0$.
Computing $\mathcal{H}_n(G)$ requires all graph maps into $G$ from
$I_1^{\square (n-1)}$, $I_1^{\square n}$, and $I_1^{\square (n+1)}$.
We store a graph map $f \colon G \to H$ as an array of length $|G_V|$ with entries
in $H_V$, the $i$-th entry recording the image of the $i$-th vertex of $G$.
Since a graph map is a set map preserving the edge relation, we enumerate all set
maps and test each one.
Writing the vertices of $I_1^{\square n}$ as $x_1, \dots, x_{2^n}$ and a candidate
map as $A = [v_1, \dots, v_{2^n}]$, the test loops over all
$(x_i, x_j) \in (I_1^{\square n})_E$ and checks that $(v_i, v_j) \in G_E$.

\paragraph{Removing the degeneracies.}
Recall that $A$ is degenerate when $\delta_i^- A = \delta_i^+ A$ for some $i$, so
we first need the faces of $A$.
These are computed by looping over the vertices $v$ of $I_1^{\square (n-1)}$ and
appending $\delta_i^\varepsilon A(v)$ to a list.
We then pass through the generated set of $n$-cubes and, for each $A$ and each
$1 \leq i \leq n$, compare $\delta_i^- A$ with $\delta_i^+ A$, discarding $A$ as
soon as the two agree.

\paragraph{Building the boundary matrices.}
The previous two steps produce bases for $\mathcal{C}_{n+1}(G)$,
$\mathcal{C}_{n}(G)$, and $\mathcal{C}_{n-1}(G)$.
For each basis element $A$ of $\mathcal{C}_k(G)$ we build a sparse vector
representing $\partial A$, initialized to zero.
For each $1 \leq i \leq k$ and each $\varepsilon \in \{-, +\}$ we compute
$\delta_i^\varepsilon A$ and search for it in the stored list of non-degenerate
$(k-1)$-cubes.
If it is absent it is degenerate, hence zero in $\mathcal{C}_{k-1}(G)$, and we skip
it; if it occurs in position $j$ we add $\pm 1$ to the $j$-th entry, with the sign
given by the boundary formula.
Taking these vectors as columns gives a sparse representation of $\partial_k$,
which we form for $k = n$ and $k = n+1$.
Computing the two ranks then yields $\dim_F \mathcal{H}_n(G; F)$.

The cost of this procedure is dominated by the first step.
Generating the $(n+1)$-cubes alone examines $|G_V|^{2^{n+1}}$ set maps, which puts
$\mathcal{H}_3(C_5)$ at an estimated $5 \times 10^5$ seconds and everything beyond
it out of reach; the full figures are given in \cref{sec:experiments}.
Improving the generation of cubes is therefore the natural place to begin.

\section{Generating singular $n$-cubes} \label{sec:generation}
The naive algorithm scales poorly in both the size of the graph and the degree, and the generation of graph maps is the reason.
Since $I_1^{\square (n+1)}$ has $2^{n+1}$ vertices, computing $\mathcal{H}_n(G)$ by the method of \cref{sec:naive} tests $|G_V|^{2^{n+1}}$ set maps.
This section replaces that enumeration.

\paragraph{Pairing lower cubes to form higher ones.}
The procedure of \cref{sec:naive} accepts any graph as domain, and so makes no use of the structure of $I_1^{\square n}$.
The following observation does.

\begin{proposition}
    Let $G_0 = I_1^{\square n}$, with vertices relabelled as $(x_1, \dots, x_n, 0)$, and $G_1 = I_1^{\square n}$, with vertices relabelled as $(x_1, \dots, x_n, 1)$. Define the set 
    \[ E = \{((x_1, \dots ,x_{n}, \varepsilon), (x_1, \dots ,x_{n},1-\varepsilon)) \ | \ x_i, \varepsilon \in \{0,1\}\} \text{,}\]
    and let $G$ to be the graph such that $G_V = {G_0}_V \sqcup {G_1}_V$, and $G_E = {G_0}_E \sqcup {G_1}_E \sqcup E$.
    Then $G = I_1^{\square n+1}$. 
    \qed
\end{proposition}
Two copies of $I_1^{\square n}$ are glued to form $I_1^{\square (n+1)}$, the first becoming its $(n+1)$-st negative face and the second its $(n+1)$-st positive face.
This is the box product with $I_1$ in different terminology.

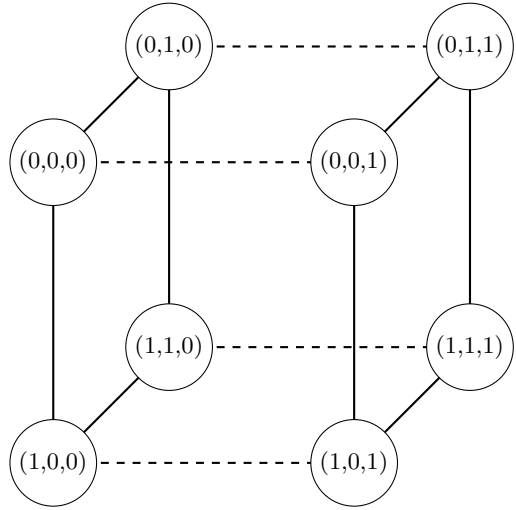
\begin{figure}[H]
\centering
\begin{tikzpicture}
\coordinate (A) at (0,0,0);
\coordinate (B) at (4,0,0); 
\coordinate (C) at (4,4,0);
\coordinate (D) at (0,4,0);
\coordinate (E) at (0,0,4);
\coordinate (F) at (4,0,4);
\coordinate (G) at (4,4,4);
\coordinate (H) at (0,4,4);

\draw[thick] (H) -- (D) -- (A) -- (E) -- cycle;

\draw[thick] (G) -- (C) -- (B) -- (F) -- cycle;

\draw[dashed, thick] (H) -- (G);
\draw[dashed, thick] (D) -- (C);
\draw[dashed, thick] (A) -- (B);
\draw[dashed, thick] (E) -- (F);

\node[circle, draw, fill=white, minimum size=25pt, inner sep=2pt,  font=\small] at (A) {(1,1,0)}; 
\node[circle, draw, fill=white, minimum size=25pt, inner sep=2pt, font=\small] at (B) {(1,1,1)};
\node[circle, draw, fill=white, minimum size=25pt, inner sep=2pt, font=\small] at (C) {(0,1,1)};
\node[circle, draw, fill=white, minimum size=25pt, inner sep=2pt, font=\small] at (D) {(0,1,0)};
\node[circle, draw, fill=white, minimum size=25pt, inner sep=2pt, font=\small] at (E) {(1,0,0)};
\node[circle, draw, fill=white, minimum size=25pt, inner sep=2pt, font=\small] at (F) {(1,0,1)};
\node[circle, draw, fill=white, minimum size=25pt, inner sep=2pt, font=\small] at (G) {(0,0,1)};
\node[circle, draw, fill=white, minimum size=25pt, inner sep=2pt, font=\small] at (H) {(0,0,0)};
\end{tikzpicture}
\caption{$I_1^{\square 2}$ paired with itself to form $I_1^{\square 3}$.}
\end{figure}

The same applies to singular cubes in $G$. This way, we need to check fewer edges: if $A = [ a_1, \dots, a_{2^n} ]$ and $B = [ b_1, \dots, b_{2^n} ]$, pairing them yields an $(n+1)$-cube precisely when $(a_i,b_i) \in G_E$ for every $i = 1, \dots, 2^n$. 
If $A$ and $B$ form an $(n+1)$-cube, we denote this cube by $A * B$.
We write $\texttt{is\_pair\_n\_cube}(A,B)$ for the resulting test. 
If the test succeeds, $A*B$ is an $(n+1)$-cube with $\delta_{n+1}^-(A*B)=A$ and $\delta_{n+1}^+(A*B)=B$.
To generate all $n$-cubes we run this test over all pairs in $\mathcal{L}_{n-1}(G) \times \mathcal{L}_{n-1}(G)$. 

Every $n$-cube has a unique $n$-th negative face and a unique $n$-th positive face, so this produces each $n$-cube exactly once from the $(n-1)$-cubes.
Rather than generating the $(n-1)$-, $n$-, and $(n+1)$-cubes separately, we therefore build them in sequence, starting from the $0$-cubes.

Generating the $(n+1)$-cubes now requires $|\mathcal{L}_n(G)|^2$ tests, each of which is itself cheaper, comparing $2^n$ pairs of vertices rather than the $n2^{n}$ edges of $I_1^{\square (n+1)}$.
On a complete graph nothing is gained: there $|\mathcal{L}_n(G)| = |G_V|^{2^{n}}$, and the count returns to $|G_V|^{2^{n+1}}$.
The advantage appears as soon as $G$ is not complete, and grows with its sparsity, since $|\mathcal{L}_n(G)|$ is then far smaller than $|G_V|^{2^{n}}$.
This is visible in \cref{sec:experiments}, where the gain is largest for $T^3$ and absent for $K_{10}$.

\paragraph{Precomputing vertex neighborhoods.}

\begin{definition}
    Given a graph $G$ and a vertex $v \in G_V$, the \textit{neighborhood} of $v$ is defined to be the set $N(v) = \{w \in G_V | (v,w) \in G_E \}$, i.e. the set of vertices $v$ is connected to. 
\end{definition}

Note that by this definition, $v \in N(v)$. 
Testing $(v,w) \in G_E$ by scanning the edge set is wasteful when the same test is repeated as often as it is inside $\texttt{is\_pair\_n\_cube}$.
We therefore precompute a dictionary of all neighborhoods once, before any maps are generated, and replace each edge test by a membership test in $N(v)$.
This is a well-studied idea, cf.\ \cite{chiba-nishizeki-adjacency,valiente:adjacency}.

\paragraph{Not storing the ($n+1$)-cubes.}
The $n$- and $(n-1)$-cubes must be kept: they generate the cubes above them, and the non-degenerate ones index the boundary matrices.
The $(n+1)$-cubes do neither.
Their only role is to contribute a column to $\partial_{n+1}$, so as soon as a pairing is found to be a valid cube we test it for degeneracy.
If it is not, we compute its boundary against the already-generated non-degenerate $n$-cubes, store the resulting coordinate vector, and discard the cube itself.

The gain here is in memory rather than time.
A coordinate vector is sparse, with at most $2(n+1)$ non-zero entries, whereas the cube it came from is an array of $2^{n+1}$ vertices.

\section{Degeneracies and generation of matrices} \label{sec:degeneracies}

Pairing does more than speed up generation: because a cube now arrives together with its two top faces, it also carries information that the naive algorithm has to recover by computation.

\paragraph{Keeping track of degenerate coordinates.}
The naive algorithm finds the degenerate cubes by computing every face of every cube.
The following observation makes that unnecessary.

\begin{proposition} \label{prop:degens}
Let $A$ and $B$ be singular $n$-cubes in a graph $G$ such that $A * B$ is an
$(n+1)$-cube.
For $1 \leq i \leq n$, the cube $A*B$ is degenerate in the coordinate $i$ if and
only if both $A$ and $B$ are.
It is degenerate in the coordinate $n+1$ if and only if $A = B$.
\end{proposition}

\begin{proof}
    It is easy to check that for $1 \leq i \leq n$, $\delta_i^\varepsilon(A * B) = \delta_i^\varepsilon(A) * \delta_i^\varepsilon(B)$.
    By definition, $\delta_{n+1}^-(A * B) = A$ and $\delta_{n+1}^+(A * B) = B$. 
    The result is then immediate. 
\end{proof}
\cref{prop:degens} lets us maintain the set of degenerate coordinates as the cubes
are generated, rather than recovering it afterwards.
A cube is stored as a pair of arrays: the map itself, and its set of degenerate
coordinates.
All $0$-cubes are non-degenerate.
When $A$ and $B$ pair to an $(n+1)$-cube, its set of degenerate coordinates is the
intersection of those of $A$ and $B$, together with $n+1$ in case $A = B$.
Filtering then reduces to testing that set for emptiness.

\paragraph{Using vertex ordering for faces.} 
An $n$-cube $A$ was formed as $\delta_n^-(A) * \delta_n^+(A)$, and pairing concatenates the two lists of vertices.
Its $n$-th negative face is therefore the first half of the list and its $n$-th positive face the second half, and the same holds inductively within each half. 
Concretely, order the vertices of $I_1^{\square n}$ as $X_n = [x_1, \dots, x_{2^n}]$ by the integer $\sum_{j} x_j 2^{j-1}$, so that $x_k$ is the binary expansion of $k-1$ read least significant digit first.
The $i$-th entry of $A$ is then the image $A(x_i)$, and the faces of $A$ can be read off by index arithmetic instead of precomposition.

Given an $i$, we want to analyze the list $[x_1[i], \dots x_{2^n}[i]]$. 
Due to the ordering of $x_1 , \dots x_{2^n}$, this list will look like an alternating sequence of $2^{i-1}$ zeroes followed by $2^{i-1}$ ones, repeated $2^{n-i}$ times.
For example, $X_3=[000,100,010,110,001,101,011,111]$, and when $i=2$ we get the sequence $[0,0,1,1,0,0,1,1]$.
This gives a closed form for the faces.

Writing $\varepsilon \in \{0,1\}$ for the two signs, the $k$-th entry of
$\delta_i^\varepsilon(A)$ is therefore the entry of $A$ in position
\[
  \left\lfloor \frac{k-1}{2^{i-1}} \right\rfloor 2^{i}
  \;+\; \varepsilon\, 2^{i-1}
  \;+\; \bigl((k-1) \bmod 2^{i-1}\bigr) \;+\; 1,
  \qquad 1 \leq k \leq 2^{n-1}.
\]
The three summands locate, respectively, the block of $2^{i}$ consecutive positions
containing $k$, the half of that block selected by $\varepsilon$, and the offset
within that half.
Computing a face is thus $2^{n-1}$ index computations with no precomposition and no
search.

\paragraph{Using a dictionary for coordinates.}
Converting a boundary to a coordinate vector requires locating each face among the non-degenerate cubes one degree down. 
In the naive algorithm this is a linear search through the non-degenerate $(k-1)$-cubes.

As with neighborhoods, we replace the search by a dictionary, keyed by cube and
returning its index.
Each face is then a single lookup, and a missing key identifies a degenerate face,
which contributes nothing to the boundary.

\section{Quotienting by the hyperoctahedral group} \label{sec:quotient}

The graph $I_1^{\square n}$ admits an action of the $n$-th hyperoctahedral group
$R_n = S_n \ltimes (\mathbb{Z}_2)^n$, and by a result of Greene, Welker, and Wille
\cite{greene-welker-wille} the chain complex may be quotiented by the induced
action before homology is taken, without changing the answer.
Since $|R_n| = n! \cdot 2^n$, this is the largest single reduction available to us,
and it grows with the degree.

The result of \cite{greene-welker-wille} is stated for arbitrary cubical sets and
proved over $\mathbb{Q}$; it applies here because $\mathcal{C}_\bullet(G)$ is the
normalized chain complex of the singular cubical set of $G$
\cite{carranza-kapulkin:cubical-graphs}.
In degree $n$ it holds over any field whose characteristic does not divide
$(n+1)!$, so from this point on we compute $\mathcal{H}_n$ over $\mathbb{F}_p$ with
$p$ the smallest prime greater than $n+1$.
We state everything in the language of graphs.
For that, recall that the vertices of $I_1^{\square n}$ are $(I_1^{\square n})_V = \{(x_1, \dots, x_n) \ | \ x_i \in \{0,1\}  \}$.

\begin{definition} \label{def:rev,symm} ${}$
    \begin{enumerate}
        \item A \textit{generating reversal} is a graph map $\rho_i^n \colon I_1^{\square n} \to I_1^{\square n}$ by $\rho_i^n(x_1, \dots , x_n ) = (x_1, \dots, x_{i-1}, 1-x_i, x_{i+1}, \dots x_n)$. For a given $n$, the \textit{$n$-th set of reversals} is the set (or group) of graph maps generated under composition by all generating reversals.

        \item Let $S_n$ be the $n$-th symmetric group. For an element $\tau \in S_n$, define a map $T_\tau \colon I_1^{\square n} \to I_1^{\square n}$ by $T_\tau (x_1,\dots, x_n) = (x_{\tau(1)}, \dots x_{\tau(n)})$. 
        We call the set (or group) of all such maps the \textit{$n$-th set of symmetries}.
    \end{enumerate}
\end{definition}

We then have the expected:

\begin{proposition} \label{prop:hyperoct}
    The automorphism group $\textsf{Aut}(I_1^{\square n})$ of $I_1^{\square n}$ is the $n$-th hyperoctahedral group $S_n \ltimes (\mathbb{Z}_2)^n$, and it is generated by the $n$-th set of reversals and $n$-th set of symmetries. \qed
\end{proposition}

By the semidirect product decomposition of \cref{prop:hyperoct}, every $\sigma \in R_n$ is uniquely a pair $(\tau, r)$ with $\tau \in S_n$ and $r = (r_1, \dots, r_n)$, $r_i \in \{1,-1\}$; we let $\sigma$ act on cubes by precomposition, $\sigma(A) = A \circ \sigma$.
Here, $r_i = 1$ means that the reversal part is constant in that coordinate, and $r_i = -1$ means that it is not.

\begin{definition}
    The sign of a permutation $\tau$, denoted $\textsf{sgn}(\tau)$, is  $(-1)^{N(\tau)}$ where $N(\tau)$ is the number of inversions.
    For $\sigma = (\tau, r) \in R_n$, the \textit{sign} of $\sigma$ is defined by $\textsf{sgn}(\sigma)=\textsf{sgn}(\tau) \cdot \prod_{i=1}^n r_i$.
\end{definition}

Alternatively, one can define the sign of $(\tau, r)$ as the determinant of the matrix in the natural representation of the hyperoctahedral group.
We do not pursue it here, since this point of view does not aid the computation.

The content of the theorem below is that $A$ and $\textsf{sgn}(\sigma)\sigma(A)$ may be identified.

\begin{definition} \label{def:sub_ch_cpx}
    Given a graph $G$, define the subcomplex $\mathcal{R}_\bullet (G) \subseteq \mathcal{C}_\bullet(G)$ by 
    \[ \mathcal{R}_n(G) =\left\langle \{A - \textsf{sgn}(\sigma) \sigma(A) \ | \ A \in \mathcal{C}_n(G), \sigma \in R_n\} \right\rangle \text{.} \]
\end{definition}

\begin{theorem}
    Let $G$ be a graph, $n \in \mathbb{N}$, and $p$ a prime not dividing $(n+1)!$. 
    Let $\mathcal{C}_\bullet(G)$ and $\mathcal{R}_\bullet(G)$ be as above. 
    Then $H_n(\mathcal{C}_\bullet(G);\mathbb{F}_p) \cong H_n(\mathcal{C}_\bullet(G)/\mathcal{R}_\bullet(G);\mathbb{F}_p)$.
\end{theorem}
\begin{proof}
    This can be found in \cite{greene-welker-wille}.
\end{proof}

We refer to the cosets as equivalence classes under $A \sim \textsf{sgn}(\sigma)\sigma(A)$, and compute the boundary of one representative per class. 
Since $|R_n| = n! \cdot 2^n$, this removes up to that factor from the number of boundaries computed, and the saving grows rapidly with the degree.
We need to modify the following parts of the existing algorithm:
\begin{itemize}
    \item generation of (equivalence classes of) maps
    \item removal of degenerate cubes
    \item boundary matrix computation
\end{itemize}

\paragraph{Generation of maps.}
What we need is a way to generate equivalence classes directly, and the class of the identity map $I_1^{\square n} \to I_1^{\square n}$ determines all the others. 
Maps are arrays of vertices, so if the identity map is $[v_1, \dots, v_{2^n}]$ then each $\textsf{sgn}(\sigma)\sigma(\id)$ is a permutation $[v_{m_1}, \dots, v_{m_{2^n}}]$ of it, together with the sign of $\sigma$. 
Now, for a graph $G$ and a map $I_1^{\square n} \to G$ represented by $[w_1, \dots w_{2^n}]$, the element of the equivalence class corresponding to the same $\sigma \in R_n$
will be the same permutation $[w_{m_1}, \dots, w_{m_{2^n}}]$. 
Storing these permutations and signs once, for the whole group, lets us produce the class of any map on demand.

Generation now produces lists of equivalence classes rather than lists of maps, and pairs classes $\mathcal{A}$ and $\mathcal{B}$ rather than individual cubes, by testing the representative of $\mathcal{A}$ against every element of $\mathcal{B}$.
We only need to consider the representative of $\mathcal{A}$, say $A$, since if $A * B$ forms an $(n+1)$-cube for some $B \in \mathcal{B}$, $\sigma(A) * \sigma(B)$ also forms an $(n+1)$-cube and will be in the equivalence class of $A * B$. 
We also only need to consider pairing $\mathcal{A}$ with $\mathcal{B}$ and not $\mathcal{B}$ with $\mathcal{A}$, since if $A * B$ is an $(n+1)$-cube then $B * A$ already lies in its equivalence class.
Once we find a pair $A*B$, we add the equivalence class of $A*B$ to an ongoing list of equivalence classes. 
Note that we only store representatives, and generate full equivalence classes on demand with the previously described method.

The one drawback is that the same class can be produced twice with different representatives, so a deduplication pass is required. 

\paragraph{Removal of degenerate cubes.}
Let $A$ be an $n$-cube in $G$. 
Suppose that for some $\sigma \in R_n$ with $\textsf{sgn}(\sigma) = -1$ that $\sigma(A) = A$. 
Then $A - (-\sigma(A)) = 2A \in \mathcal{R}_n(G)$. 
Since we are working over a field of characteristic not 2 (assuming $n \geq 1$), we have $A \in \mathcal{R}_n(G)$. 
Thus a cube invariant under some $\sigma$ of negative sign is already zero in the quotient, and may be discarded exactly as a degenerate cube is. 
We call cubes of this form \textit{semi-degenerate}.
In fact, if $A$ is degenerate in a coordinate $i$, then $\delta_i^-(A) = \delta_i^+(A)$, so $A$ is invariant under $\sigma = (\id, r)$ with $r_i = -1$ and $r_j = 1$ for $j \neq i$. 
Since the sign of this reversal is $-1$, all degenerate cubes are semi-degenerate.
Thus, we can remove all semi-degenerate cubes, and hence remove all degenerate cubes at the same time.

\paragraph{Boundary matrix computation.}
Only the coordinate dictionary changes.
The keys will be all $n$-cubes that are in some equivalence class, and the values need to be a tuple $(i,sgn)$, where $sgn$ is the sign of the element in its equivalence class.
We then get a column vector per equivalence class of maps, and use the faces of the representative $A$ to compute the boundary. 
If the value of $\delta_i^\varepsilon(A)$ is $(k,sgn)$ in the dictionary, we set $v[k]=(-1)^i*sgn$ if $\varepsilon = -$ and $v[k]=(-1)^{i+1}*sgn$ if $\varepsilon = +$.

\section{Preprocessing graphs} \label{sec:preprocessing}

Everything so far has taken the graph as given.
This section replaces it: we look for a smaller graph $G'$, usually a subgraph, with the same discrete homology, and compute there instead.
Because the cost of the computation grows so steeply in $|G_V|$, even a modest reduction is worth a great deal.
Results of this kind are known for path homology \cite{grigoryans}, but they do not transfer, since discrete and path homology are known to differ \cite{barcelo-greene-jarrah-welker:comparison}; we prove analogues directly.

We first recall the relevant notion of homotopy equivalence.

\begin{definition}
    Let $G$ and $G'$ be graphs and $f,g \colon G \to G'$ graph maps.
    A \textit{homotopy} $H \colon f \to g$ is a choice of $n \in \mathbb{N}$ and a map $H \colon G \square I_n \to G'$ such that $H(-,0) = f$ and $H(-,n) = g$.
\end{definition}

\begin{definition}
    Let $G$ and $G'$ be graphs. 
    A map $f \colon G \to G'$ is a \textit{homotopy equivalence} if there exists a map $g \colon G' \to G$ along with a pair of homotopies $gf \to \id[G]$ and $fg \to \id[G']$.
    For graphs $G$ and $G'$, if such an $f$ exists, we say $G$ and $G'$ are \textit{homotopy equivalent}.
\end{definition}

The motivation for these definitions in our context is the following theorem:

\begin{theorem}[{\cite[Thm.~3.4]{barcelo-greene-jarrah-welker:comparison}}] \label{thm:homotopy-invariance}
    If $G$ and $G'$ are homotopy equivalent graphs, and $R$ a commutative ring with unity, $\mathcal{H}_n(G;R) \cong \mathcal{H}_n(G';R)$ for all $n \in \mathbb{N}$.
\end{theorem}

So it suffices to find a smaller graph homotopy equivalent to $G$.
Recall that $N(v)$ denotes the neighborhood of $v$, and that $v \in N(v)$ under our conventions. 
\begin{definition}
    Let $G$ be a graph and $v \in G_V$. 
    We say that $v$ is \textit{homotopically removable} if there exists a vertex $w \neq v$ such that $N(v) \subseteq N(w)$.
\end{definition}
Equivalently, $v$ has a neighbor $w$ adjacent to everything $v$ is adjacent to.

As the name suggests: 
\begin{theorem} \label{TH:htpy_rem_v}
    Let $G$ be a graph and $R$ a commutative ring with unit. Suppose $v \in G_V$ is homotopically removable. 
    Define the graph $G'$ to be the subgraph of $G$ induced by $G'_V = G_V \setminus \{v\}$, i.e. the graph $G$ with the vertex $v$ removed. 
    Then $\mathcal{H}_n(G;R) \cong \mathcal{H}_n(G';R)$ for all $n \in \mathbb{N}$.
\end{theorem}
\begin{proof}
    We construct a homotopy equivalence between $G$ and $G'$.
    Let $i \colon G' \hookrightarrow G$ be the inclusion.
    Suppose $w$ is a vertex with $N(v) \subseteq N(w)$, and define $f \colon G \to G'$ by
    \[
      f(x) = \begin{cases}
        w & \text{if } x = v, \\
        x & \text{otherwise.}
      \end{cases}
    \]
    This is a graph map since $N(v) \subseteq N(w)$. 
    We have that $f \circ i = \id[G']$. 
    It remains to construct a homotopy $H \colon i \circ f \to \id[G]$.
    Set $H \colon G \square I_1 \to G$ to be $H(-,0) = i \circ f$ and $H(-,1) = \id[G]$. 
    All we need is that this is a valid graph map, which is true if $H(x,0) \sim H(x,1)$ for all $x \in G_V$. 
    If $x \neq v$ this is clearly true since then $H(x,0) = H(x,1)$.
    If $x = v$ we have that $H(v,0) = w \sim v = H(v,1)$. 
    Thus, by \cref{thm:homotopy-invariance}, $G$ and $G'$ are homotopy equivalent and have the same homology groups.
\end{proof}

These are the \emph{dominated} vertices of \cite{barmak-et-al-dominated}, and \cref{TH:htpy_rem_v} is the discrete-homotopy analogue of the reductions underlying edge collapse for flag complexes \cite{boissonnat-siddharth-edge-collapse}.
The combinatorial condition is identical; only the invariance theorem differs.
Detecting all such vertices costs a single pass over the vertex set, so the reduction is worth attempting unconditionally; \cref{sec:experiments} reports its effect.

Vertices of degree $2$ admit a finer analysis, which we can carry out completely in degrees $0$ and $1$ and only conjecturally above.
The obstruction is reminiscent of the suspension problem for graphs, a well-known open problem in discrete homotopy theory; we do not pursue the connection here.
Throughout, a \emph{vertex of degree $2$} is one with $|N(v)| = 3$, namely $v$ itself and two others. 
Suppose $v$ is of degree 2 and connected to two other distinct vertices $x$ and $y$.
Define $\textsf{Conn}_{xy}$ to be the minimum length of a path in $G$ to get from $x$ to $y$, without going through the vertex $v$. If no such path exists, set $\textsf{Conn}_{xy} = -1$. 

For the remainder of this section, let $G$ be a graph, and $v$ a vertex of degree 2 connected to distinct vertices $x$ and $y$.
We have a few cases based on what $\textsf{Conn}_{xy}$ is: 

\begin{theorem}
    Suppose $\textsf{Conn}_{xy} = -1$, that is, there is no path from $x$ to $y$ that does not go through $v$. 
    Define $G'$ to be the subgraph of $G$ induced by removing the vertex $v$. Then $\mathcal{H}_0(G;R) \oplus R \cong \mathcal{H}_0(G';R)$ and $\mathcal{H}_1(G;R) \cong \mathcal{H}_1(G';R)$.
\end{theorem}
\begin{proof}
    Since $\mathcal{H}_0(G;R) \cong \bigoplus_{i \in I} R$, where $I$ is the set of connected components of $G$, and the vertices $x$ and $y$ are not connected except through $v$, removing $v$ causes $G'$ to have one more connected component than $G$, so $\mathcal{H}_0(G;R) \oplus R \cong \mathcal{H}_0(G';R)$.

For the claim about $\mathcal{H}_1$, we use the van Kampen theorem for graphs, found in \cite[Thm.~5.5]{kapulkin-mavinkurve}.
    We begin by defining two subgraphs of $G$: $G_x$ and $G_y$.
    By assumption, $x$ and $y$ belong to different connected components of $G'$.
    The graph $G_x$ is the induced subgraph of $G$ containing the vertex $v$ and all connected components of $G'$ except the connected component of $y$.
    Let $G_y$ be the induced subgraph of $G$ containing the vertex $v$ and the connected component of $y$ in $G'$.
    We have a pushout: 
    \begin{center}
\begin{tikzcd}
I_0 \arrow[d, "v"'] \arrow[r, "v"] & G_x \arrow[d, hook] \\
G_y \arrow[r, hook]                & G                  
\end{tikzcd}
\end{center}
Where the maps $v \colon I_0 \to G_x$ and $v \colon I_0 \to G_y$. 
Write $A_1$ for the discrete fundamental group of a pointed graph.
Every map $C_4 \to G$ whose image is a $3$-cycle or a $4$-cycle factors through $G_x \hookrightarrow G$ or through $G_y \hookrightarrow G$, so the van Kampen theorem \cite[Thm.~5.5]{kapulkin-mavinkurve} applies.
As the pushout is taken over the trivial group, it gives the free product $A_1(G) \cong A_1(G_x) * A_1(G_y)$.
Since $\mathcal{H}_1$ is the abelianization of $A_1$ \cite[Thm.~4.1]{barcelo-capraro-white}, and abelianization carries free products to direct sums, $\mathcal{H}_1(G) \cong \mathcal{H}_1(G_x) \oplus \mathcal{H}_1(G_y)$.

Let $G'_x$ and $G'_y$ denote $G_x$ and $G_y$ with the vertex $v$ removed.
In each of $G_x$ and $G_y$ the vertex $v$ is homotopically removable, so $\mathcal{H}_1(G_x) \cong \mathcal{H}_1(G'_x)$ and $\mathcal{H}_1(G_y) \cong \mathcal{H}_1(G'_y)$ by \cref{TH:htpy_rem_v}.
Moreover $G' = G'_x \sqcup G'_y$, and discrete homology preserves coproducts, so
\[\mathcal{H}_1(G) \cong \mathcal{H}_1(G_x) \oplus \mathcal{H}_1(G_y) \cong \mathcal{H}_1(G'_x) \oplus \mathcal{H}_1(G'_y) \cong \mathcal{H}_1(G')\text{,}\]
which completes the proof.
\end{proof}

We believe the result for $\mathcal{H}_1$ to be true for all higher homology groups, but the proof remains elusive: 

\begin{conjecture}
    Suppose $\textsf{Conn}_{xy} = -1$, that is, there is no path from $x$ to $y$ that does not go through $v$. 
    Define $G'$ to be the subgraph of $G$ induced by removing the vertex $v$. Then $\mathcal{H}_0(G;R) \oplus R \cong \mathcal{H}_0(G';R)$ and $\mathcal{H}_n(G;R) \cong \mathcal{H}_n(G';R)$ for all $n \geq 1$.
\end{conjecture}

For the case of $\mathsf{Conn}_{xy} = 1$ there is an edge between $x$ and $y$, so the vertex $v$ is homotopically removable, since $N(v) \subseteq N(x)$. 
For $n=2$, we can only offer a conjecture:
\begin{conjecture}
    Suppose $\textsf{Conn}_{xy} = 2$.
    Define $G'$ to be the subgraph of $G$ induced by removing the vertex $v$. Then $\mathcal{H}_n(G;R) \cong \mathcal{H}_n(G';R)$ for all $n \geq 0$.
\end{conjecture}

The final case is $\mathsf{Conn}_{xy} \geq 3$ with the following theorem: 
\begin{theorem}
    Suppose $\textsf{Conn}_{xy} \geq 3$.
    Define $G'$ to be the subgraph of $G$ induced by removing the vertex $v$. Then $\mathcal{H}_0(G;R) \cong \mathcal{H}_0(G';R)$ and $\mathcal{H}_1(G;R) \cong \mathcal{H}_1(G';R) \oplus R$.
\end{theorem}
\begin{proof}
    For $\mathcal{H}_0$, note that removing $v$ does not affect the number of connected components. 
    For $\mathcal{H}_1$, we use another van Kampen argument. 
    To this end, let $m = \mathsf{Conn}_{xy}$ and let $I_{m} \to G'$ be a path between $x$ and $y$.
    Define a map $I_m \to C_{m+2}$ to map a vertex $i \in (I_m)_V$ to $i \in (C_{m+2})_V$.
    (As in \cref{ex:graphs}, the vertices of $I_m$ are labelled $0$, $1$, \ldots, $m$, and the vertices of $C_{m+2}$ are labelled $0$, $1$, \ldots, $m+1$.) 
    Finally, define the inclusion $C_{m+2} \hookrightarrow G$ to map vertices $0$ to $m$ to the path from $x$ to $y$, and $m+1$ to $v$. 
    We then have the pushout: 
    \begin{center}
    \begin{tikzcd}
    I_m \arrow[d] \arrow[r] & G' \arrow[d, hook] \\
    C_{m+2} \arrow[r, hook]                & G                  
    \end{tikzcd}
\end{center}
So by the van Kampen theorem \cite[Thm.~5.5]{kapulkin-mavinkurve}, this pushout is preserved by the discrete fundamental group, and therefore by its abelianization, yielding $\mathcal{H}_1(G;R) \cong \mathcal{H}_1(G';R) \oplus R$, since $\mathcal{H}_1(C_{m+2}; R) \cong R$.
\end{proof}

Again, we have a conjecture regarding the general case which has yet to be proven.

\begin{conjecture}
    Suppose $\textsf{Conn}_{xy} \geq 3$.
    Define $G'$ to be the subgraph of $G$ induced by removing the vertex $v$. Then $\mathcal{H}_n(G;R) \cong \mathcal{H}_n(G';R)$ for $n \geq 2$.
\end{conjecture}

One approach towards the proof of this conjecture would require a higher-dimensional version of the van Kampen theorem for homology, i.e., the Mayer--Vietoris sequence.
At present, such a theorem has not been shown in the context of discrete homotopy theory, however a restricted version, insufficient for our purposes, was considered in \cite[Proof of Thm.~5.2]{barcelo-greene-jarrah-welker:vanishing}.

\section{Experimental evaluation} \label{sec:experiments}

This section reports the effect of each improvement on the running time.
We compare four implementations: the naive algorithm of \cref{sec:naive}; the
algorithm with the inductive generation of cubes of \cref{sec:generation} together
with the improvements of \cref{sec:degeneracies}, which we call the
\emph{non-quotient} algorithm; the same with the quotient of \cref{sec:quotient},
which we call the \emph{quotient} algorithm; and finally a parallel implementation
of the last.
The naive and non-quotient algorithm both compute over $\mathbb{F}_2$, and the quotient algorithm computes over $\mathbb{F}_p$, where $p$ is the smallest prime larger than $n+1$.

Our code is available at
\url{https://github.com/nkershaw01/Discrete_Cubical_Homology}.
It is written in Julia and was run on the Digital Research Alliance of Canada
Narval cluster on a single node with 240~GB of memory and, for the parallel
implementation, 64 threads.
The test graphs are $C_5$, the Greene sphere $G^{sph}$, the discrete $3$-torus
$T^3$, the star-shaped $C_5^{star}$, and the complete graph $K_{10}$, all
introduced in \cref{sec:homology}.
We do not apply preprocessing in these tests, since on this set it is either
decisive or trivial: $K_{10}$ collapses to a point and $C_5^{star}$ to $C_5$,
while the remaining graphs have no dominated vertices.
Its effect is discussed separately at the end of the section.

Where a computation did not terminate in the available time, we report an estimate,
obtained by sampling and extrapolating, and mark it with an asterisk.
Estimates should be read as orders of magnitude rather than as predictions.

\subsection*{Running times}

\cref{fig:times} summarizes all four implementations across all five graphs, and
\cref{tab:times} gives the underlying numbers.

\begin{figure}[htbp]
  \centering
  \begin{tikzpicture}
  \begin{groupplot}[
      group style={group size=3 by 2, horizontal sep=1.6cm, vertical sep=1.6cm,
                   xlabels at=edge bottom, xticklabels at=edge bottom,
                   ylabels at=edge left},
      width=0.36\textwidth, height=4.2cm,
      ymode=log, xmin=-0.3, xmax=4.3, xtick={0,1,2,3,4},
      xlabel={degree $n$}, ylabel={seconds},
      label style={font=\scriptsize}, tick label style={font=\tiny},
      title style={font=\small}, ymajorgrids, log basis y=10,
    ]
    \nextgroupplot[title={$C_5$}, ytick={1e-4,1e0,1e4,1e8,1e12,1e16}]
  \addplot[nvcol,mark=*,thick] coordinates {(0,0.000158) (1,0.00607) (2,2.75)};
  \addplot[nvcol,mark=*,thick,dashed,mark options={fill=white}] coordinates {(2,2.75) (3,500000) (4,9e+16)};
  \addplot[ngcol,mark=square*,thick] coordinates {(0,4.88e-05) (1,0.000593) (2,0.0201) (3,21.1)};
  \addplot[ngcol,mark=square*,thick,dashed,mark options={fill=white}] coordinates {(3,21.1) (4,1e+08)};
  \addplot[nqcol,mark=triangle*,thick] coordinates {(0,7.76e-05) (1,0.000385) (2,0.00967) (3,5.47)};
  \addplot[nqcol,mark=triangle*,thick,dashed,mark options={fill=white}] coordinates {(3,5.47) (4,1e+07)};
  \addplot[npcol,mark=diamond*,thick] coordinates {(0,0.000638) (1,0.00171) (2,0.00447) (3,0.637) (4,141000)};
    \nextgroupplot[title={$G^{sph}$}, ytick={1e-4,1e2,1e8,1e14,1e20,1e26}]
  \addplot[nvcol,mark=*,thick] coordinates {(0,0.000574) (1,0.0608) (2,440)};
  \addplot[nvcol,mark=*,thick,dashed,mark options={fill=white}] coordinates {(2,440) (3,3e+10) (4,3e+26)};
  \addplot[ngcol,mark=square*,thick] coordinates {(0,0.00017) (1,0.00214) (2,0.237)};
  \addplot[ngcol,mark=square*,thick,dashed,mark options={fill=white}] coordinates {(2,0.237) (3,700000) (4,3e+11)};
  \addplot[nqcol,mark=triangle*,thick] coordinates {(0,0.000148) (1,0.00115) (2,0.072) (3,123)};
  \addplot[nqcol,mark=triangle*,thick,dashed,mark options={fill=white}] coordinates {(3,123) (4,2e+09)};
  \addplot[npcol,mark=diamond*,thick] coordinates {(0,0.000738) (1,0.00191) (2,0.0132) (3,13.8)};
    \nextgroupplot[title={$T^3$}, ytick={1e-3,1e10,1e23,1e36,1e49,1e62}]
  \addplot[nvcol,mark=*,thick] coordinates {(0,1.4) (1,20300)};
  \addplot[nvcol,mark=*,thick,dashed,mark options={fill=white}] coordinates {(1,20300) (2,2e+12) (3,1e+29) (4,4e+62)};
  \addplot[ngcol,mark=square*,thick] coordinates {(0,0.00415) (1,0.124) (2,38.7)};
  \addplot[ngcol,mark=square*,thick,dashed,mark options={fill=white}] coordinates {(2,38.7) (3,4e+06) (4,7e+19)};
  \addplot[nqcol,mark=triangle*,thick] coordinates {(0,0.00509) (1,0.142) (2,30.2)};
  \addplot[nqcol,mark=triangle*,thick,dashed,mark options={fill=white}] coordinates {(2,30.2) (3,1e+06) (4,8e+17)};
  \addplot[npcol,mark=diamond*,thick] coordinates {(0,0.00253) (1,0.0198) (2,1.83) (3,8140)};
    \nextgroupplot[title={$C_5^{star}$}, ytick={1e-4,1e2,1e8,1e14,1e20,1e26}]
  \addplot[nvcol,mark=*,thick] coordinates {(0,0.000987) (1,0.0991) (2,936)};
  \addplot[nvcol,mark=*,thick,dashed,mark options={fill=white}] coordinates {(2,936) (3,4e+10) (4,5e+26)};
  \addplot[ngcol,mark=square*,thick] coordinates {(0,0.000149) (1,0.00218) (2,0.532)};
  \addplot[ngcol,mark=square*,thick,dashed,mark options={fill=white}] coordinates {(2,0.532) (3,4e+06) (4,2e+17)};
  \addplot[nqcol,mark=triangle*,thick] coordinates {(0,0.000157) (1,0.00147) (2,0.14) (3,1510)};
  \addplot[nqcol,mark=triangle*,thick,dashed,mark options={fill=white}] coordinates {(3,1510) (4,1e+13)};
  \addplot[npcol,mark=diamond*,thick] coordinates {(0,0.000793) (1,0.00211) (2,0.0292) (3,176)};
    \nextgroupplot[title={$K_{10}$}, ytick={1e-4,1e6,1e16,1e26,1e36,1e46}]
  \addplot[nvcol,mark=*,thick] coordinates {(0,0.00152) (1,0.811)};
  \addplot[nvcol,mark=*,thick,dashed,mark options={fill=white}] coordinates {(1,0.811) (2,100000) (3,3e+16) (4,2e+46)};
  \addplot[ngcol,mark=square*,thick] coordinates {(0,0.000377) (1,0.0485)};
  \addplot[ngcol,mark=square*,thick,dashed,mark options={fill=white}] coordinates {(1,0.0485) (2,50000) (3,3e+16) (4,2e+46)};
  \addplot[nqcol,mark=triangle*,thick] coordinates {(0,0.000312) (1,0.0215)};
  \addplot[nqcol,mark=triangle*,thick,dashed,mark options={fill=white}] coordinates {(1,0.0215) (2,20000) (3,1e+11) (4,1e+30)};
  \addplot[npcol,mark=diamond*,thick] coordinates {(0,0.000763) (1,0.0102) (2,51.9)};
  \nextgroupplot[hide axis, xmin=0, xmax=1, ymin=1, ymax=10,
      legend style={at={(0.5,0.5)}, anchor=center, draw=none, fill=none,
                    font=\small, cells={anchor=west}, legend columns=1}]
    \addlegendimage{nvcol,mark=*,thick}\addlegendentry{naive}
    \addlegendimage{ngcol,mark=square*,thick}\addlegendentry{non-quotient}
    \addlegendimage{nqcol,mark=triangle*,thick}\addlegendentry{quotient}
    \addlegendimage{npcol,mark=diamond*,thick}\addlegendentry{parallel quotient}
    \addlegendimage{black,thick,dashed,mark=none}\addlegendentry{extrapolated}
  \end{groupplot}
  \end{tikzpicture}
  \caption{Time to compute $\mathcal{H}_n$ for each implementation.
    Note the logarithmic vertical scale, which spans more than sixty orders of
    magnitude for $T^3$.
    Solid segments are measured; dashed segments are extrapolated.}
  \label{fig:times}
\end{figure}

\renewcommand{\arraystretch}{1.15}
\begin{table}[htbp]
  \centering
  \begin{threeparttable}
  \begin{tabular}{llccccc}
    \toprule
    \textbf{Graph} & \textbf{Implementation} & $\mathcal{H}_0$ & $\mathcal{H}_1$ & $\mathcal{H}_2$ & $\mathcal{H}_3$ & $\mathcal{H}_4$ \\
    \midrule
    $C_5$ & naive          & $1.58 \times 10^{-4}$ & $6.07 \times 10^{-3}$ & $2.75$ & $5 \times 10^{5}$\tnote{*} & $9 \times 10^{16}$\tnote{*} \\
          & non-quotient   & $4.88 \times 10^{-5}$ & $5.93 \times 10^{-4}$ & $2.01 \times 10^{-2}$ & $21.1$ & $1 \times 10^{8}$\tnote{*} \\
          & quotient       & $7.76 \times 10^{-5}$ & $3.85 \times 10^{-4}$ & $9.67 \times 10^{-3}$ & $5.47$ & $1 \times 10^{7}$\tnote{*} \\
          & parallel       & $6.38 \times 10^{-4}$ & $1.71 \times 10^{-3}$ & $4.47 \times 10^{-3}$ & $0.637$ & $1.41 \times 10^{5}$ \\
    \midrule
    $G^{sph}$ & naive        & $5.74 \times 10^{-4}$ & $6.08 \times 10^{-2}$ & $440$ & $3 \times 10^{10}$\tnote{*} & $3 \times 10^{26}$\tnote{*} \\
              & non-quotient & $1.70 \times 10^{-4}$ & $2.14 \times 10^{-3}$ & $0.237$ & $7 \times 10^{5}$\tnote{*} & $3 \times 10^{11}$\tnote{*} \\
              & quotient     & $1.48 \times 10^{-4}$ & $1.15 \times 10^{-3}$ & $7.20 \times 10^{-2}$ & $123$ & $2 \times 10^{9}$\tnote{*} \\
              & parallel     & $7.38 \times 10^{-4}$ & $1.91 \times 10^{-3}$ & $1.32 \times 10^{-2}$ & $13.8$ & --- \\
    \midrule
    $T^3$ & naive        & $1.40$ & $2.03 \times 10^{4}$ & $2 \times 10^{12}$\tnote{*} & $1 \times 10^{29}$\tnote{*} & $4 \times 10^{62}$\tnote{*} \\
          & non-quotient & $4.15 \times 10^{-3}$ & $0.124$ & $38.7$ & $4 \times 10^{6}$\tnote{*} & $7 \times 10^{19}$\tnote{*} \\
          & quotient     & $5.09 \times 10^{-3}$ & $0.142$ & $30.2$ & $1 \times 10^{6}$\tnote{*} & $8 \times 10^{17}$\tnote{*} \\
          & parallel     & $2.53 \times 10^{-3}$ & $1.98 \times 10^{-2}$ & $1.83$ & $8.14 \times 10^{3}$ & --- \\
    \midrule
    $C_5^{star}$ & naive        & $9.87 \times 10^{-4}$ & $9.91 \times 10^{-2}$ & $936$ & $4 \times 10^{10}$\tnote{*} & $5 \times 10^{26}$\tnote{*} \\
                 & non-quotient & $1.49 \times 10^{-4}$ & $2.18 \times 10^{-3}$ & $0.532$ & $4 \times 10^{6}$\tnote{*} & $2 \times 10^{17}$\tnote{*} \\
                 & quotient     & $1.57 \times 10^{-4}$ & $1.47 \times 10^{-3}$ & $0.140$ & $1.51 \times 10^{3}$ & $1 \times 10^{13}$\tnote{*} \\
                 & parallel     & $7.93 \times 10^{-4}$ & $2.11 \times 10^{-3}$ & $2.92 \times 10^{-2}$ & $176$ & --- \\
    \midrule
    $K_{10}$ & naive        & $1.52 \times 10^{-3}$ & $0.811$ & $1 \times 10^{5}$\tnote{*} & $3 \times 10^{16}$\tnote{*} & $2 \times 10^{46}$\tnote{*} \\
             & non-quotient & $3.77 \times 10^{-4}$ & $4.85 \times 10^{-2}$ & $5 \times 10^{4}$\tnote{*} & $3 \times 10^{16}$\tnote{*} & $2 \times 10^{46}$\tnote{*} \\
             & quotient     & $3.12 \times 10^{-4}$ & $2.15 \times 10^{-2}$ & $2 \times 10^{4}$\tnote{*} & $1 \times 10^{11}$\tnote{*} & $1 \times 10^{30}$\tnote{*} \\
             & parallel     & $7.63 \times 10^{-4}$ & $1.02 \times 10^{-2}$ & $51.9$ & --- & --- \\
    \bottomrule
  \end{tabular}
  \begin{tablenotes}
    \item[*] estimated. \quad --- : not attempted.
  \end{tablenotes}
  \end{threeparttable}
  \caption{Time in seconds to compute $\mathcal{H}_n$.}
  \label{tab:times}
\end{table}

The naive algorithm reaches degree $2$ on the three smallest graphs and nothing
beyond; even $\mathcal{H}_2(T^3)$ is out of reach, at an estimated
$2 \times 10^{12}$ seconds, or some sixty thousand years.
The inductive generation of cubes is responsible for the largest single gain: it
reduces the degree-$2$ times by factors between $137$ and $1.86 \times 10^3$ on
$C_5$, $G^{sph}$, and $C_5^{star}$, and brings $\mathcal{H}_2(T^3)$ down from that
estimate to $38.7$ seconds.
$K_{10}$ is the exception throughout, and for the expected reason: on a complete
graph every set map is a graph map, so there is nothing for the generation step to
prune.

The quotient by the hyperoctahedral group behaves as the size $n! \cdot 2^n$ of the
group suggests, contributing little in low degrees and increasingly much in high
ones.
In degree $2$ it gives factors between $1.28$ and $3.8$; in degree $3$ it takes
$G^{sph}$ from an estimated $7 \times 10^{5}$ seconds to a measured $123$, and
$C_5^{star}$ from an estimated $4 \times 10^{6}$ to a measured
$1.51 \times 10^{3}$.
In other words, it is the step that converts two of our degree-$3$ estimates into
actual computations.

Parallelization has real overhead, and we report it rather than suppress it.
In degree $0$, and for three of the five graphs in degree $1$, the parallel
implementation is \emph{slower} than the sequential one, by factors of up to $8$;
the work per thread is simply too small to repay the cost of distributing it.
From degree $2$ onward it pays, with factors between $2.2$ and $16.5$ in degree
$2$ and around $8.6$ to $9$ in degree $3$.
Its main effect is at the frontier: $\mathcal{H}_4(C_5)$ becomes computable, in
$1.41 \times 10^{5}$ seconds, or roughly $39$ hours.

\subsection*{Where the time goes}

The rank computations are not the bottleneck at any stage.
\cref{tab:setup} splits the running time into the construction of the boundary
matrices and the computation of their ranks.
The times do not sum exactly to those of \cref{tab:times}, since they were obtained
on separate runs.

\begin{table}[htbp]
  \centering
  \begin{tabular}{lcccccc}
    \toprule
    & \multicolumn{2}{c}{naive} & \multicolumn{2}{c}{non-quotient} & \multicolumn{2}{c}{quotient} \\
    \cmidrule(lr){2-3}\cmidrule(lr){4-5}\cmidrule(lr){6-7}
    \textbf{Group} & setup & rank & setup & rank & setup & rank \\
    \midrule
    $\mathcal{H}_2(C_5)$     & $2.58$ & $3.74 \times 10^{-3}$ & $1.83 \times 10^{-2}$ & $2.80 \times 10^{-3}$ & $8.75 \times 10^{-3}$ & $9.47 \times 10^{-8}$ \\
    $\mathcal{H}_2(G^{sph})$ & $443$  & $0.226$               & $0.229$               & $4.81 \times 10^{-2}$ & $8.42 \times 10^{-2}$ & $5.68 \times 10^{-5}$ \\
    $\mathcal{H}_2(T^3)$     & ---    & ---                   & $29.3$                & $9.37$                & $30.6$                & $4.36 \times 10^{-3}$ \\
    $\mathcal{H}_3(C_5)$     & ---    & ---                   & $15.0$                & $8.80$                & $4.77$                & $6.51 \times 10^{-3}$ \\
    \bottomrule
  \end{tabular}
  \caption{Time in seconds spent building the boundary matrices versus computing
    their ranks.}
  \label{tab:setup}
\end{table}

Under the naive algorithm the setup dominates by three orders of magnitude, which
is what motivated \cref{sec:generation,sec:degeneracies}.
Once those improvements are in place the ranks become a visible share of the
total --- a third of the time for $\mathcal{H}_2(T^3)$ and $\mathcal{H}_3(C_5)$ ---
and one might expect rank computation to become the next bottleneck.
The quotient removes it again, and by a much larger margin than it reduces the
setup time: for $\mathcal{H}_2(T^3)$ the setup is essentially unchanged while the
rank computation falls by a factor of over $2000$.
The reason is that the quotient discards semi-degenerate cubes, which include all
degenerate ones, so the matrices it produces are not merely smaller but also
better conditioned for elimination.
This is why we did not pursue a parallel rank computation.

\subsection*{Overall improvement}

\cref{tab:speedups} gives the ratio of the naive time to the parallel quotient time.

\renewcommand{\arraystretch}{1.15}
\begin{table}[htbp]
  \centering
  \begin{threeparttable}
  \begin{tabular}{lccccc}
    \toprule
    \textbf{Graph} & $\mathcal{H}_0$ & $\mathcal{H}_1$ & $\mathcal{H}_2$ & $\mathcal{H}_3$ & $\mathcal{H}_4$ \\
    \midrule
    $C_5$ & $0.25$ & $3.55$ & $615$ & $7.85 \times 10^{5}$\tnote{*} & $6.38 \times 10^{11}$\tnote{*} \\
    $G^{sph}$ & $0.78$ & $31.8$ & $3.33 \times 10^{4}$ & $2.17 \times 10^{9}$\tnote{*} & --- \\
    $T^3$ & $553$ & $1.03 \times 10^{6}$ & $1.09 \times 10^{12}$\tnote{*} & $1.23 \times 10^{25}$\tnote{*} & --- \\
    $C_5^{star}$ & $1.24$ & $47.0$ & $3.21 \times 10^{4}$ & $2.27 \times 10^{8}$\tnote{*} & --- \\
    $K_{10}$ & $1.99$ & $79.5$ & $1.93 \times 10^{3}$\tnote{*} & --- & --- \\
    \bottomrule
  \end{tabular}
  \begin{tablenotes}
    \item[*] at least one of the two times is an estimate.
  \end{tablenotes}
  \end{threeparttable}
  \caption{Speedup of the parallel quotient algorithm over the naive algorithm.}
  \label{tab:speedups}
\end{table}

The gains grow with the degree, which is the point: every improvement here targets
the doubly exponential growth in $n$, and the ratios accordingly run from below $1$
in degree $0$ to $10^5$ and beyond in degree $3$.
They also grow with the sparsity of the graph, since pairing prunes nothing on a
complete graph.
$K_{10}$ is therefore the worst case and $T^3$, the largest but sparsest of our
examples, the best.

\subsection*{Preprocessing}

Preprocessing is excluded from the tables above because on this test set its effect
is all-or-nothing, but it is worth quantifying.
Removing dominated vertices collapses $K_{10}$ to a single vertex, making every
computation immediate, and converts $C_5^{star}$ into $C_5$.
For $\mathcal{H}_3(C_5^{star})$ this replaces a $176$ second computation with a
$0.637$ second one, a factor of $276$, at the cost of a single pass over the
vertex set.
The remaining three graphs have no dominated vertices, and preprocessing is a
no-op.

This is the behaviour to expect in general.
Preprocessing is cheap enough to be worth attempting unconditionally, and its payoff is governed entirely by how many dominated vertices the graph has, which is to say, it helps most on precisely the locally dense graphs where the rest of the algorithm struggles.

\subsection*{Validation}

The groups themselves are not the point of these experiments; the test graphs
were chosen because their homology is known, so that the output can be checked,
but it is worth recording that the checks pass.
Wherever two or more of the four implementations terminate on the same input, they
return the same answer, and the results agree with the values in
\cref{ex:homology_groups} and with those recorded in
\cite{barcelo-greene-jarrah-welker:comparison}.
The quotient algorithm over $\mathbb{F}_p$ and the auxiliary implementation over
$\mathbb{Q}$ described in \cref{sec:summary} likewise agree on every input on which
both terminate.

\section{Summary and open problems} \label{sec:summary}

Machine computation of discrete homology has reached degree $2$ for graphs on a
handful of vertices, and no further.
The algorithm presented here computes degree-$3$ groups for every graph in our test
set and the degree-$4$ group of the $5$-cycle, with speedups over the naive
algorithm running from $10^3$ to $10^{12}$ in degree $2$ and from $10^8$ to
$10^{25}$ in degree $3$ (\cref{tab:speedups}).

Generating $n$-cubes by pairing $(n-1)$-cubes, rather than filtering set maps, is
responsible for most of this, and for the rest of it indirectly, since the other
improvements depend on it.
On its own it gives factors of $10^2$ to $10^3$ in degree $2$ on the sparse
examples, and more than $10^{10}$ for $T^3$.
Carrying degeneracy data through that construction, and reading faces off by index
arithmetic, removes most of what remains of the setup cost.
Quotienting by the hyperoctahedral group contributes little below degree $2$ and a
great deal above it, converting our degree-$3$ estimates for $G^{sph}$ and
$C_5^{star}$ into actual computations and all but eliminating the rank computation
as a cost.
Preprocessing by deleting dominated vertices is cheap enough to attempt
unconditionally and can be decisive: it collapses $K_{10}$ to a point, and reduces
$\mathcal{H}_3(C_5^{star})$ from $176$ seconds to $0.637$.

\paragraph{Limitations.}
The improvements here target growth in the degree; growth in the size and
connectivity of the graph is barely affected, and $K_{10}$ remains the hardest of
our examples in every degree.
Parallelization is not free either: in degree $0$, and for most of our graphs in
degree $1$, the parallel implementation is slower than the sequential one, and it
repays its overhead only once the work per thread is substantial.

More fundamentally, the whole approach is shaped by the fact that the chain complex
cannot be built.
The Greene sphere, a graph on ten vertices, has more than $6.4 \times 10^{12}$
singular $5$-cubes.
This is what separates the present problem from the setting in which the reduction
machinery of computational topology was developed.
Matrix reduction with clearing and twist
\cite{chen-kerber:twist,bauer-kerber-reininghaus:chunks,bauer:ripser}, Morse-theoretic
simplification of filtrations \cite{mischaikow-nanda:morse}, coreduction
\cite{mrozek-batko:coreduction}, and edge collapse
\cite{boissonnat-siddharth-edge-collapse} all simplify a complex one already
possesses, and are remarkably effective at it; every technique here is instead a way
of avoiding the complex altogether.
We do not know whether the two can be reconciled.

\paragraph{Open problems and future directions.}

\begin{enumerate}
  \item \emph{The conjectures of \cref{sec:preprocessing}.}
    Each would extend preprocessing to a case we cannot currently handle.
    A proof in the case $\mathsf{Conn}_{xy} \geq 3$ would likely require a
    Mayer--Vietoris sequence for discrete homology; no such sequence is available,
    and only a restricted version, insufficient for our purposes, is known
    \cite[Proof of Thm.~5.2]{barcelo-greene-jarrah-welker:vanishing}.

  \item \emph{Scaling in the size of the graph.}
    This is the principal obstacle, and we do not currently see how to attack it.
    Progress here would matter more than any further constant-factor improvement.

  \item \emph{Edge graphs.}
    Write $E(G)$ for the graph whose vertices are the edges of $G$, with two of them
    joined when they are opposite edges of a $4$-cycle in $G$.
    The $n$-cubes of $G$ are then the vertices of $E^n(G)$ and the $(n+1)$-cubes its
    edges, so iterating $E$ generates all the maps we need.
    Searching for $4$-cycles is faster than pairing cubes, so this would improve map
    generation.
    We have not found a way to combine it with \cref{sec:quotient}, which requires
    pairing equivalence classes rather than individual maps, and in high degrees the
    quotient is worth more.
    A construction supporting both, in particular one that does not produce
    duplicate equivalence classes, would be a substantial gain.

  \item \emph{Reduction during generation.}
    Reducing the boundary coordinate vectors as they are produced should save
    memory.
    We have not found a formulation compatible with parallelization: reducing on each
    thread separately saves little, costs time, and does not compose into a cheaper
    global reduction.

  \item \emph{Algebraic Morse theory.}
    Kozlov's algebraic Morse theory \cite{kozlov-algebraic-morse} applies to any
    based chain complex, so in principle it applies to $\mathcal{C}_\bullet(G)$, and
    the gains it brings elsewhere in computational topology are substantial
    \cite{mischaikow-nanda:morse}.
    Whether it can be made to apply in practice is another matter, since a Morse
    matching is a structure on the complex and the complex is exactly what we cannot
    afford to write down.
    A version that could be computed alongside the inductive generation of cubes,
    rather than after it, would be a genuine advance.

  \item \emph{Neighbouring theories.}
    The specialized degree-$1$ algorithms for path homology
    \cite{chowdhury-memoli:path,dey-li-wang:path} suggest a complementary line of
    attack: rather than computing $\mathcal{H}_n$ for general $n$ faster, characterize
    what a discrete homology class looks like in a fixed low degree and build an
    algorithm around that characterization.

  \item \emph{Persistence.}
    Extending these computations to filtered graphs, as required by the applications
    in \cite{kapulkin-kershaw:tda}, is the subject of forthcoming work.
\end{enumerate}

\paragraph{Code.}
The implementation is available at
\url{https://github.com/nkershaw01/Discrete_Cubical_Homology}.
It provides three entry points: the algorithm of this paper, computing over
$\mathbb{F}_p$ for the smallest prime $p > n+1$; a variant over $\mathbb{Q}$ using
Julia's built-in rank, which is faster but not guaranteed correct for
ill-conditioned matrices; and a dedicated routine for $\mathcal{H}_1$ over
$\mathbb{Z}/2$ based on the edge graph described above.

 \bibliography{all-refs.bib}

 \appendix
 \renewcommand{\thesection}{\Alph{section}}

\end{document}